\documentclass[11pt, twoside]{article}
\usepackage{amsmath,amsthm,amssymb}
\usepackage{times}
\usepackage{enumerate}

\theoremstyle{definition}
\newtheorem{thm}{Theorem}[section]
\newtheorem{cor}[thm]{Corollary}
\newtheorem{pro}[thm]{Proposition}
\newtheorem{lem}[thm]{Lemma}

\newtheorem{rem}[thm]{Remark}

\numberwithin{equation}{section}

\newcommand{\gd}{\gamma_\diamond}

\newcommand{\dd}{\delta}
\newcommand{ \bn}{{\bar n}}

\newcommand{\bbl}{Blanc, LeBris \& Lions }

\newcommand{\hn}{{\hat n}}

\newcommand{\tn}{{\tilde n}}
\newcommand{\be}{\begin{equation}}
\newcommand{\ee}{\end{equation}}
\newcommand{\ben}{\begin{equation*}}
\newcommand{\een}{\end{equation*}}
\newcommand{\ph}{\varphi}
\newcommand{\lk}{L_k}
\newcommand{\kk}{{1\over k}}
\newcommand{\e}{\varepsilon}
\newcommand{\Om}{\Omega}
\newcommand{\Ss}{\Sigma}
\newcommand{\po}{{\partial\Omega}}
\newcommand{\ol}{{\Omega\cap L}}
\newcommand{\oel}{{\Omega\cap \e L}}

\newcommand{\ZZ}{\mathbb{Z}}
\newcommand{\RR}{\mathbb{R}}
\newcommand{\RRR}{\RR^3}
\newcommand{\ZZZ}{\ZZ^3}

\newcommand{\bc}{\setminus}

\newcommand{\Omp}{{\Om^+}}
\newcommand{\Omm}{{\Om^-}}
\newcommand{\Ompm}{{\Om^\pm}}
\newcommand{\Fp}{{F^+}}
\newcommand{\Fm}{{F^-}}
\newcommand{\Fpm}{{F^\pm}}

\newcommand{\Lo}{{L}}

\newcommand{\cho} {\chi_{{}_{\Om}}}

\newcommand{\p} {\partial}
\def\Xint#1{\mathchoice
 {\XXint\displaystyle\textstyle{#1}}%
 {\XXint\textstyle\scriptstyle{#1}}%
 {\XXint\scriptstyle\scriptscriptstyle{#1}}%
 {\XXint\scriptscriptstyle\scriptscriptstyle{#1}}%
 \!\int}
\def\XXint#1#2#3{{\setbox0=\hbox{$#1{#2#3}{\int}$}
 \vcenter{\hbox{$#2#3$}}\kern-.5\wd0}}

\def\dashint{\Xint-}
\newcommand{\dint}{\dashint}
\newcommand{\fk}{\frac{1}{|K|}}
\begin{document}

\title{\textsc{The Interfacial Energy of a Phase Boundary\\ via a Lattice-Cell Average Approach}}
\author{\textsc{\textbf{Phoebus Rosakis}}\\
 Department of Theoretical \\and Applied Mathematics\\ University of Crete\\Heraklion 70013, Greece\\{\&} \\ Institute of Applied and Computational Mathematics\\
 Foundation for Research and Technology\\Heraklion 70013, Greece\\
\texttt{rosakis@uoc.gr}}

\date{}
% The correct date

\maketitle

\begin{abstract} \noindent The calculation of the discrete atomistic energy of a crystal near the continuum limit encounters difficulties caused by the geometric discrepancy between the continuum region occupied by the body, and the discrete collection of lattice points contained in it. This results in ambiguities in the asymptotic expansion of the energy for small values of the lattice parameter, that  are traced back to the lattice point problem of number theory. The lattice-cell average of the discrete energy  is introduced and is shown to eliminate this ambiguity in various circumstances. It is used to find explicit continuum expressions for surface energies and interfacial energies of coherent phase boundaries in deformed crystals in terms of the interatomic potential.\end{abstract}

\section{Introduction}
\label{intro}
The description of the atomistic energy of a crystal near the continuum limit entails various difficulties, some of which are geometric in nature. These stem from the discrepancy between the continuum region occupied by the body, and the discrete collection of lattice points contained in it. The energy in question is 
\be\label{ene}E_\e\{y,\Om\}= \frac{\e^3}{2}\sum_{x\in\Om\cap \e L} \,\sum_{\;z\in\Om\cap \e L} \Phi\left(\frac{y(z)-y(x)}{\e}\right).
\ee
Here $\Om\subset\RRR$ is the \emph{continuous body} (macroscopic reference region), $L \subset \RRR$ is a simple (Bravais) lattice, $\e>0$ is a scaling factor proportional to the lattice parameter of the scaled lattice $\e L$, $\Phi:\RRR\to\RR$ is the interatomic potential for pair interactions and $y:\Om\to\RRR$ a macroscopic deformation, which we assume is followed by all atoms. The \emph{discrete body} is $\oel$.
 \bbl (Theorem 3, \cite{blanc1}) obtain an asymptotic expansion of \eqref{ene} for sufficiently smooth diffeomorphisms $y$:
\be\label{asy}E_\e\{y,\Om\}=\int_\Om W(\nabla y )dx+\e \int_\po \Gamma(\nabla y,\po) dA+\e^2\int_\Om U(\nabla y ,\nabla \nabla y)dx+O(\e^2),\quad \e\to0.\ee
Here the volume integrands, namely the stored energy function $W$ and higher gradient energy function $U$, are explicitly determined (see \eqref{0} for $W$), but the surface energy density $\Gamma$ and certain other surface terms of $O(\e^2)$ are not. The reasons for this are geometrical.
Given the continuous body $\Om$, the volume $|\Om|$ is not in general equal to the so-called \emph{discrete volume} $\e^3\#(\Om\cap\e L)$, the number of lattice points in $\Om$ times the lattice cell volume. The asymptotics as $\e\to0$ of the difference
\be\label{rem}R(\e)=|\Om|-\e^3\#(\Om\cap\e L),\ee
 known as the \emph{remainder}, is the subject of the \emph{lattice point problem} in number theory, aspects of which are still open; e.g., Beck \& Robins \cite{br}. Even in simple cases, e.g., $\Om$ a sphere, the remainder is highly oscillatory and its asymptotic expansion in powers of $\e$ depends on the $\e$ sequence. The representation \eqref{asy} holds provided $\e=\e_k\to 0$, as $k\to\infty$, $k\in\ZZ$, where $\e_k$ is a putative sequence of scale factors such that the remainder vanishes: $R(\e_k)=0$. This condition\footnote{It does not appear to be known whether such a sequence exists in more than one dimension even in simple cases, such as $\Om$ a sphere or ellipsoid.} points to the geometrical discrepancy between the discrete and continuous bodies as the culprit for the fact that the asymptotic expansion of the energy in powers of $\e$ depends on the choice of sequence of $\e\to0$. In one dimension, this is demonstrated directly by Mora-Corral \cite{mora}, who obtains all terms up to order $\e^2$ and their deprendence on the $\e$-sequence explicitly. 
 
 In more than one dimensions, the role of the geometric discrepancy between the continuous and discrete body is emphasized by Rosakis \cite{rosakis}, who reduces the calculation of the energy \eqref{ene} to certain lattice point problems. The asymptotic behavior of these depends on whether boundary surfaces are rational (crystallographic) or irrational. In particular, some terms in the expansion depend on the $\e$-sequence in the rational case, e.g., if $\Om$ is a lattice polyhedron, but not in the irrational case, for example when $\Om$ is a smooth strictly convex domain.
 The resulting expression for the surface energy density $\Gamma$ in \eqref{asy}, of the form $\hat \gamma(\nabla y, n)$ (Proposition 4.3, \cite{rosakis}), is discontinuous at rational values of the boundary normal $n$ and continuous at irrational values. This explicit form of $\hat \gamma$ is valid for a special choice of sequence, $\e_k=1/k$, $k\in\ZZ$. 
 
 One way to eliminate the discontinuous behavior is to modify the continuous body\footnote{Given the discrete body---a finite subset $D$ of $\e L $---there is no unique way of choosing a region $\Om\subset\RRR$ (open set or its closure) that contains it and no other lattice points (i.e., $\oel=D$). Thus $\Om$ can be modified without changing $D$ or the energy.} for each $\e_k=1/k$ so that the remainder $R(1/k)=0$, namely its volume equals the discrete volume. This is done in a simple geometrical way, and the result (Proposition 5.1, \cite{rosakis}) is an explicit and continuous surface energy density function:
 \be\label{g5}\gamma(F,n)=-\frac{1}{4}\sum_{w\in \Lo} |w\cdot n|\Phi(Fw)\ee
 for any deformation gradient matrix $F$ with $\det F>0$ and any unit normal vector $n$.
However, the fact that the continuous body now depends on $\e$ and a special sequence $\e_k=1/k$ is required is still awkward.\footnote{The modified continuous body has the physically desirable property that the continuum mass and the discrete mass coincide, which is is equivalent to vanishing of the remainder \eqref{rem}. See Remark 10 in \cite{rosakis}.} 

Here we propose an alternative approach that eliminates the problems associated with geometric discrepancy between the continuous and discrete body. Typically, one is interested in sums over the discrete body $\ol$, of the form 
\be\label{sum}
 \sum_{x\in\Om\cap L} \ph(x) \quad \hbox{or}\quad \sum_{x\in\Om\cap L} \sum_{\;y\in\Om\cap L} \psi(x,y).
 \ee 
Choosing $\ph=1$ in the first sum gives the lattice point number $\#(\Om\cap L)$ of $\Om$, which also equals the mass of the body, iff each lattice point has a unit atomic mass. The energy \eqref{ene} is in the form of the second sum. 

 In general, the orientation of the lattice relative to the macroscopic (continuous) body can be determined by methods such as x-ray crystallography; see, e.g., Zachariasen \cite{zachariasen}. In contrast, the exact position of the lattice relative to the body cannot be ascertained to the same degree, and perhaps it is meaningless to do so, at least for large enough bodies. This motivates introducing the \emph{cell average} of quantities defined as sums over the discrete body, such as \eqref{sum}. That would be the average value of the sum over all possible positions (but fixed orientation) of the lattice relative to the continuous body $\Om$. The sums in \eqref{sum} are periodic in lattice translations $L\to u+L$, $u\in\RRR$, with period the lattice cell $K$ (defined below in \eqref{k}). Thus we let the cell average of the single sum be
\be\label{c1}\dint_K \sum_{x\in\Om\cap (u+L)} \ph(x)d u:=\frac{1}{|K|}\int_K \sum_{x\in L} \cho(x+u)\ph(x+u)d u,\ee
where we have extended $\ph:\Om\to\RR$ to the whole of $\RRR$ arbitrarily, and $\cho$ is the characteristic function of $\Om$. The cell average of the double sum is defined analogously; replace $L$ in both sums in \eqref{sum} by $u+L$ and take the average over $u\in K$, see \eqref{c2} below. Essentially, the cell average is an ensemble average over all possible positions of the lattice relative to the continuous body.

The advantages are immediate. For example, consider a scaled version of the single sum in \eqref{sum}, the Riemann sum $\e^3 \sum_{x\in\Om\cap\e L} \ph(x)$. It approaches $\int_\Om\ph$ as $\e\to 0$, but the higher order terms in its expansion in powers of $\e$ are difficult to characterize explicitly for general $\Om$ and $\Phi$, as discussed by Guillemin \& Sternberg \cite{gs}. On the other hand, it is very easy to show that \emph{the cell average of the Riemann sum is exactly equal to the integral $\int_\Om\ph$ for all }$\e>0$ (see Lemma \ref{lem1} below). Once averaged, the lattice point problem becomes trivial: the cell average of the discrete volume equals the continuous volume regardless of scale, and there are no problematic higher order terms. This is encouraging, and we study the cell average of the energy.

In Section \ref{sec2} we define cell averages of single and multiple sums on the discrete body $\Om\cap\e L$, in particular, the energy \eqref{ene} of a deformed crystal subject to binary interactions. In effect, this replaces one of the sums by an integral over $\Om$, and this is advantageous for analytical calculations.
The cell average \eqref{e1} of the energy is now the integral over the continuous body (instead of the sum over the discrete body) of an energy density that depends of finite differences of the deformation; see Corollary \ref{cor2}.

Section \ref{sec3} contains a calculation of the continuum limit of the cell-averaged energy in case the deformation is merely Lipschitz, so that it admits sharp phase boundaries (gradient discontinuity surfaces). In the limit we recover the usual elastic energy, without additional interfacial energies; these are of higher order in $\e$ and we study them later on in Section \ref{sec5}.

Turning to smooth deformations in Section \ref{sec4}, we study the term of order $\e$ in the asymptotic expansion of the energy; this is identical to the surface energy \eqref{g5} obtained by Rosakis (Proposition 5.1, \cite{rosakis}).

Coherent phase or twin boundaries are modelled in Section \ref{sec5} as continuous deformations with piecewise constant gradient that jumps across a plane; the interface is thus sharp. In some situations, atomic displacements due to deformation twins is some crystals appear to be well described by such piecewise affine deformations (down to the atomic scale; an example is shown in Fig. 4(b) of Zhu, Liao \& Wu \cite{zhu}). This suggests that the Cauchy-Born hypothesis may be appropriate in the case of some twin interfaces, as atoms appear to follow an affine deformation on either side of the interface, whereas near a free surface there is microscopic relaxation from the macroscopic deformation. The energy of a piecewise affine deformation contains an interfacial term of order $\e$ like the surface energy. In one dimension, such a term was obtained by Mora-Corral (\cite{mora} Theorem 6). The result of cell averaging is an explicit representation of the interfacial energy (Proposition \ref{prop3}) in the form
$$\e\int_\Ss \sigma(F^+,F^-, \hn)dA,$$
where $\Ss$ is the interface with unit normal $ \hn$ and $F^\pm$ are the limiting values of $\nabla u$ (which has a jump across $\Ss$) on either side of $\Ss$.
It is noteworthy that the interfacial energy density is expressible completely in terms of the surface energy density $\gamma(F,n)$ of \eqref{g5} 
$$
 \sigma(F^+,F^-, \hn)=\gamma(F^+, \hn)+\gamma(F^-, \hn)-2\int_{0}^{1}\gamma\left(tF^++(1-t)F^-, \hn\right)dt.$$

We remark that all of the results involving cell averages are free from any sequential depencence issues; they are valid for an arbitrary sequence of $\e\to 0$.

The fully discrete counterparts of the aforementioned surface and interfacial energies are constructed in Section \ref{sec6}. The discrete surface energy density depends of the way the continuous body is defined (Rosakis \cite{rosakis} Proposition 5.1). For example, if $\Om$ is a lattice polyhedron, which means that its vertices are lattice points, the discrete surface energy differs from the cell-averaged one, while if $\Om$ is altered so at to eliminate the geometric discrepancy \eqref{rem} between the continuous and discrete volumes, the resulting surface energy equals the cell average (Proposition \ref{prop4}).
 Such a modification involves translating the facets of $\po$, and is not feasible for a planar phase boundary, which is interior to $\Om$. The discrete interfacial energy involves a finite sum which turns out to equal the trapezoidal approximation of an integral that appears in the corresponding cell-averaged expression. It is then possible to define an approximating sequence of interfaces, whose Miller indices grow unbounded as the original orientation is approached. The discrete surface and interfacial energies associated with this sequence approach the cell average in the limit.

%Higher order terms
%Multibody

 \section{The Cell Average Approach}\label{sec2}

Let $L\subset\RR^3$ be a simple (Bravais) lattice. If $e_i\in L$ form a lattice basis for $L$, the set 
\be\label{k}K=\Bigl\{x\in\RR^3 :x=\sum_{i=1}^3 z_ie_i,\; 0\le z_i<1\Bigr\}\ee
is a \emph{lattice cell}. For subsets $P$, $Q$ of $\RR^3$, define the Minkowski sum $P\oplus Q=\{p+q:p\in P,\; q\in Q\}$ and write $p+Q=\{p\}\oplus Q$. For $\e\in\RR$, $\e P=\{\e x: x\in P\}$.

The continuous body is a bounded Lipschitz domain $\Om\subset\RRR$, the \emph{discrete body} is $\Om\cap L$. Consider sums over the discrete body, of the form 
\eqref{sum}. We introduce the \emph{cell average} of a sum over the discrete body. That would be the average value of the sum over all possible positions of the lattice relative to the continuous body $\Om$. The sums in \eqref{sum} are periodic in lattice translations $L\to u+L$, $u\in\RRR$, with period the lattice cell $K$. Thus we define the cell average of the single sum by \eqref{c1}.
%\be\label{c1}\dint_K \sum_{x\in\Om\cap (u+L)} \ph(x)d u=\frac{1}{|K|}\int_K \sum_{x\in L} \cho(x+u)\ph(x+u)d u\ee
%where we have extended $\ph:\Om\to\RR$ to the whole of $\RRR$ arbitrarily. 
%
The cell average of the double sum is taken to be
\be\label{c2}\dint_K\sum_{x\in\Om\cap (u+L)} \sum_{\;\;y\in\Om\cap (u+L)} \psi(x,y)\;du= \fk
\int_K\sum_{x\in L} \sum_{\;y\in L}\cho(x+u)\cho(y+u) \psi(x+u,y+u)\; du.\ee

\begin{lem}\label{lem1} (i) For $\ph\in L^1(\Om)$,
\be\label{s1}\dint_K \sum_{x\in\Om\cap (u+L)} \ph(x)d u=\fk\int_\Om\ph(x)dx.\ee
(ii) If $x\mapsto\psi(x,w+x)$ is in $ L^1(\Om\cap(-w+\Om))$ for $w\in L$,
\be\label{xy}\dint_K\sum_{x\in\Om\cap (u+L)} \sum_{\;\;y\in\Om\cap (u+L)} \psi(x,y)\;du=
 \fk
 \sum_{w\in L}\int_\Om \cho(x+w)\psi(x,x+w)\; dx.\ee
 \end{lem}
 \begin{proof} Note that all sums in the right hand sides of \eqref{c1} and \eqref{c2} only have a finite number of nonzero terms because 
 \be\label{chop}|w|>\hbox{diam}\Om\implies \cho(x+w)=0 \quad \forall x\in\Om .\ee
 The integral in the right hand side of \eqref{c1} equals
 $$ \sum_{x\in L} \int_K\cho(x+u)\ph(x+u)d u=\sum_{x\in L} \int_{x+K}\cho(z)\ph(z)d z=\int_{\RR^3}\cho(z)\ph(z)d z$$
since $K$ tiles $\RR^3$; this shows (i).

The integral in the right hand side of \eqref{c2} equals
$$\sum_{x\in L} \int_K \sum_{\;y\in L}\cho(x+u)\cho(y+u) \psi(x+u,y+u)\; du=\sum_{x\in L} \int_K \sum_{\;w\in L}\cho(x+u)\cho(w+x+u) \psi(x+u,w+x+u)\; du$$
letting $w=y-x$. Proceeding as in (i), the above equals
$$ \sum_{w\in L}\sum_{\;x\in L} \int_K\cho(x+u)\cho(w+x+u) \psi(x+u,w+x+u)\; du= \sum_{w\in L}\sum_{\;x\in L} \int_{x+K}\cho(z)\cho(w+z) \psi(z,w+z)\; dz$$
$$=\sum_{w\in L}\int_{\RR^3}\cho(z)\cho(w+z) \psi(z,w+z)\; dz$$
which confirms (ii) after dividing by $|K|$. \end{proof}

\begin{cor}\label{cor1}
 The cell average of the lattice point number of $\Om$,
 \be\dint_K \#[\Om\cap (u+L)]du=\fk|\Om|\ee
 \end{cor}
 \begin{proof} Choose $\ph=\cho$ in \eqref{s1}.
\end{proof}

\begin{rem}\label{rem1}This easy result is significant, as it implies that the remainder associated with the cell-averaged lattice point problem vanishes identically. The sequence dependence in the asymptotic expansion of the energy is due to the remainder, so the result suggests that the cell-averaged energy might be free of this problem. This is confirmed by our results later on, for terms of order $\e$. One can show that $O(\e^2)$ terms are sequence-independent as well.
\end{rem}

\begin{cor}\label{cor2}
 The cell average of the energy \eqref{ene} is
 \be\label{e1} \bar E_\e\{y,\Om\}=\frac{1}{2}\sum_{w\in L}\int_\Om \cho(x+\e w)\Phi\left(\frac{y(x+\e w)-y(x)}{\e}\right)dx. \ee
 \end{cor}

\begin{rem}\label{rembad} One of the two summations over $\Om\cap\e L$ in \eqref{ene} is replaced by integration over $\Om$. It is tempting, but inappropriate, to perform two successive lattice-cell averages, one for each sum in \eqref{ene}. Instead, by translating $L$ to $u+L$ in both sums concurrently and averaging over $u$, we are leaving bond vectors---differences between lattice points----invariant. Recall that the energy depends on bond vectors. \end{rem}
\begin{rem}\label{remgood}The averaged energy \eqref{e1} has a dual continuum-discrete character. It is no longer a sum over the discrete body, but an integral over the continuous body, of an energy density; however, prior to taking the limit, the energy density depends on finite differences of the deformation, not derivatives, in a discrete fashion. \end{rem}

 \section{Lipschitz Deformations}
\label{sec3}
In order to ensure global invertibility of $y$, we will always assume that for some $ \lambda>0$,
 \be\label{inv} |y(z)-y(x)|\ge \lambda |z-x|\quad\forall x,z\in\Om,
 \ee
as in \bbl \cite{blanc1}. Suppose now that all atoms in the discrete body follow a prescribed deformation which we assume to be a Lipschitz homeomorphism of the continuous body. Observe that \eqref{inv} implies that $y^{-1}$ is Lipschitz as well. Hence $y$ is bi-Lipschitz, or 
\be\label{lh} y\in C^{0,1}(\bar\Om,\RR^3),\qquad y^{-1}\!\in C^{0,1}( y(\bar \Om),\RR^3)\ee 
The interatomic (pair) potential is $\Phi:\RR^3\to\RR$; we assume it satisfies the following:
\be\label{ph}\hbox{For each $ a>0$,}\quad \Phi\in C^0(\RR^3\bc B_a(0),\RR),\quad |\Phi(z)|<C(a) |z|^{-(3+p)},\quad |z|>a, \;\;C(a)>0, \; p>0.\ee
We set $\Phi(0)=0$ to avoid having to write $x\neq z$ in lattice sums. The scaled lattice is $\e L$ where $\e$ will approach zero. We assume $|K|=1$. 
If we define the stored energy function from the Cauchy-Born formula,
\be\label{0}W(F)=\frac{1}{2}\sum_{w\in \Lo} \Phi(Fw)\ee
for all invertible $3\times 3$ matrices $F$, then formally at least, we immediately obtain the limit of the cell-averaged energy to be \eqref{lim} below. Proceeding more carefully, we have
\begin{pro}\label{prop1}Suppose the deformation $y$ is a bi-Lipschitz homeomorphism and the interatomic potential $\Phi$ satisfies \eqref{ph}. Then the cell-averaged energy satisfies
\be\label{lim}\lim_{\e\to 0} \bar E_\e\{y,\Om\}=\int_\Om W(\nabla y(x))dx.\ee
\end{pro}
 \begin{proof} 
 Because of \eqref{inv} or the second of \eqref{lh}, $|x+\e w-x|<\lambda^{-1} |y(x+\e w)-y(x)|$ for all $\e>0$ and for all $x\in\Om\cap(-\e w+\Om)$. Hence, letting $g_\e(x,w)$ be the integrand in \eqref{e1}, we have in view of \eqref{ph}
 \be\label{bound}|g_\e(x,w)|<C\left|\frac{y(x+\e w)-y(x)}{\e}\right|^{-(3+p)}<C|\lambda w|^{-(3+p)},\ee
 where $C=C(\lambda )$ (see \eqref{ph}) because $|y(x+\e w)-y(x)|/\e\ge\lambda |w|\ge \lambda $ for all $w\in L\backslash\{0\}$. Also,
 \be\label{con}\sum_{w\in L\bc\{0\}} |w|^{-(3+p)}=M<\infty\ee 
 by Lemma 3.2 of \cite{rosakis}. By \eqref{chop}, the sum in \eqref{e1} can be placed inside the integral, since it equals the finite sum over 
 $\{w\in L: |w|\le \e^{-1}\hbox{diam}\Om\}$. Therefore,
$$ \bar E_\e\{y,\Om\}=\frac{1}{2}\int_\Om \sum_{w\in L}g_\e(x,w)dx;$$
the series within the integral above converges uniformly in $\e$ by \eqref{bound}, \eqref{con}. This ensures that
 \be\label{li}\lim_{\e\to 0} \sum_{w\in L}g_\e(x,w)= \sum_{w\in L}\lim_{\e\to 0} g_\e(x,w)=\sum_{w\in L}\Phi(\nabla u(x)w)\ee for a.e. $x\in\Om$ by the Rademacher theorem, in view of \eqref{lh}. Also $ | \sum_{w\in L}g_\e(x,w)|<C\lambda^{3+p} M $ on $\Om$ by \eqref{bound}, \eqref{con}, hence by bounded convergence, $\lim_{\e\to 0}2\bar E_\e=\int_\Om \lim_{\e\to 0}\sum_{w\in L}g_\e$, The result follows from \eqref{li}, \eqref{0}.
\end{proof}

\begin{rem}\label{rem2} The proof is simpler and the hypotheses on $u$ and $\Om$ weaker than those of Theorem 1 in \cite{blanc1}. In particular $u$ is merely Lipschitz, so phase boundaries are allowed, while $\Om$ is only required to be a bounded measurable domain. On the other hand Blanc \& Le Bris \cite{blanc2} only need $y\in W^{1,p}(\Om)$ for $p>3$ and $y^{-1}$ Lipschitz to obtain the limit \eqref{lim} for the discrete energy. \end{rem}

 \section{Smooth Deformations. Surface Energy.}
\label{sec4}
In order to avoid some cumbersome technical arguments we suppose that $\Phi$ has finite (but otherwise arbitrary) range. In particular, for some $R>0$, 
 \be\label{phr} \Phi\in C^1(\RR^3\bc B_a(0),\RR)\;\; \hbox{for each $ a\in (0,R)$};\quad\Phi(z)=0\;\; \hbox{for } |z|>R;\quad \Phi(-z)=\Phi(z) \;\; \hbox{for } z\in\RR^3.
 \ee
 The last is consistent with the rotational invariance of two-body potentials. 
 \begin{rem} All of our results can be extended to the case of infinite range, provided $\Phi$ and its first two derivatives decay fast enough. It is cumbersome and not particularly illuminating to perform this extension.
 \end{rem}
\begin{pro}\label{prop2}Suppose $\bar\Om\subset\RR^3$ is diffeomordiffeomorphicphic to a closed polyhedron or a sphere, $y$ is a $C^2(\bar\Om,\RR^3)$ diffeomorphism and $\Phi$ satisfies \eqref{phr}.
 Then the cell-averaged energy satisfies
\be\label{1}\bar E_\e\{y,\Om\}=\int_\Om W(\nabla y )dx+\e \int_\po \gamma(\nabla y,n) dA+o(\e)\quad \text{as } \e\to0,\ee
for any null $\e$-sequence, where the stored energy function $W$ is given by \eqref{0} and the surface energy density $\gamma$ by
\be\label{2}\gamma(F,n)=-\frac{1}{4}\sum_{w\in \Lo} |w\cdot n|\Phi(Fw).\ee
 \end{pro}
 \begin{proof} 
Fix $w\in L$ and let $\Om_\e=-\e w+\Om$. Then the sum in \eqref{e1} has terms of the form
\be\label{term}
E(w,\e)=\int_{\Om_\e\cap\Om}\Phi\left(\frac{y(x+\e w)-y(x)}{\e}\right)dx
\ee
Letting $\psi(x,w,\e)$ equal the integrand above for $\e>0$, define $\psi(x,w,0)=\Phi(\nabla u(x)w)$ for $x\in\Om$. Then \eqref{phr} and \eqref{inv} ensure that $\psi$ can be extended to be $C^1$ on $\bar\Om\times[0,\e_0]$ for some $\e_0>0$. Let $\p_\pm\Om=\{x\in\po: \pm w\cdot n(x)> 0\}$ and $\p_+\Om_\e=\po_\e\cap\Om$, $\p_-\Om_\e=\po\cap\bar\Om_\e$. Thinking of $\e$ as time, $\Om_\e\cap\Om$ is an evolving domain with boundary comprising $\p_\pm\Om_\e$. The boundary velocity equals $-w$ on $\p_+\Om_\e$ and $0$ on $\p_-\Om_\e$.
The transport theorem then furnishes (primes indicate derivatives with respect to $\e$)
\be\label{d1}E'(w,\e)=\frac{d}{d\e}\int_{\Om_\e\cap\Om}\psi(x,w,\e)dx=\int_{\Om_\e\cap\Om}\psi'(x,w,\e)dx-\int_{\p_+\Om_\e}\psi(x,w,\e)w\cdot n\,dA\ee
Since the velocity $-w$ is inward on $\p_+\Om$, every $x$ in its relative interior satisfies $x-\e w\in \p_+\Om_\e$ for sufficiently small $\e>0$ (while for a.e. $x\in \p_-\Om$ one has $x-\e w\not \in \Om$ for $\e>0$ sufficiently small). Moreover if $\p_0\Om=\{x\in\po:w\cdot n(x)=0\}$ the contribution to the second integral in \eqref{d1} of $(-\e w+\p_0\Om)\cap\Om\subset\p_+\Om_\e$ vanishes because $w\cdot n=0$ on this set. It follows easily that the second integral approaches $\int_{\p_+\Om}\psi(x,w,0)w\cdot n dx$ as $\e\to 0$. Thus
$$\lim_{\e\to0+}E'(w,\e)= \int_{\Om}\psi'(x,w,0+)dx-\int_{\p_+\Om}\psi(x,w,0)w\cdot n\,dA$$
A direct calculation shows that $\psi'(x,w,0+)=\nabla\Phi(\nabla u(x)w)\cdot\nabla^2y(x)[w,w]$. This is odd in $w$, since from the last of \eqref{phr}, $\nabla\Phi(Fw)$ is odd in $w$. Hence $\sum_{w\in L}\psi'(x,w,0+)=0$, and
\begin{gather*}\sum_{w\in L}E'(w,0+)=-\sum_{w\in L}\int_{\p_+\Om}\psi(x,w,0)w\cdot n\,dA=\\ -\sum_{w\in L}\int_{\p\Om}\psi(x,w,0)\langle w\cdot n\rangle_+dA=-\frac{1}{2}\sum_{w\in L}\int_{\p\Om}\psi(x,w,0)
| w\cdot n | \, dA,
\end{gather*} 
letting $\langle t\rangle_+=\max\{t,0\}$ for $t\in \RR$, and since $\psi(x,w,0+)$ is even in $w$. Dividing by 2 and invoking \eqref{2}, \eqref{0}, proves the result, noting that $E(0+)=\int_\Om\Phi(\nabla u(x)w)dx$.
\end{proof}

\begin{rem} The cell-averaged surface energy with density \eqref{2} is the same as in Rosakis (\cite{rosakis}, Proposition 5.1), but the present computation is far simpler and does not require modified domains. The cell-averaged surface energy of Proposition \ref{prop2} is valid for arbitrary sequences of $\e\to 0$.\end{rem}

 \section{Cell-Averaged Interfacial Energy of a Coherent Phase or Twin Boundary}\label{sec5}

 We calculate the (cell-average of) the energy of a deformation involving a phase boundary, which we model as a surface of discontinuity of the deformation gradient. The fully discrete energy is obtained later on in Section \ref{sec6}.
 \begin{rem} In general, microscopy reveals that actual twinning deformations sometimes appear to be close to piecewise affine (with a sharp interface at the atomic level) and sometimes seem to be better described by a smooth transition that extends over a few interplanar spacings. We refer to Ball and Mora-Corral \cite{ball}, who propose and study a continuum energy that allows both sharp and smooth interfaces. Here we only consider sharp interfaces, such as the example shown in Fig. 4(b) of Zhu, Liao \& Wu \cite{zhu}.
 \end{rem}
We suppose that $\Om$ is as in Proposition \ref{prop2} with the origin in its interior. The phase boundary $\Ss=\{x\in\Om: x\cdot \hn=0\}$ is planar with unit normal $ \hn$. It separates $\Om$ into regions $\Om^\pm=\{x\in\Om : \pm x\cdot \hn > 0\}$ with nonempty interiors. We let $\hat y$ be a piecewise homogeneous deformation, which seems to describe adequately some cases of mechanical twinning in crystals with a Bravais lattice (Zanzotto \cite{zanzotto}):
 \be\label{pw}
\hat y(x)=\begin{cases}
 F^+ x, & \text{ } x\in\Om^+, \\
 F^- x, & \text{ } x\in\Om^-\cup \Ss,\\
 \end{cases}
 \ee
 where the constant matrices $F^\pm$ are related by the Hadamard compatibility condition 
 \be\label{hd} F^+=F^-+a\otimes \hn\ee
for some constant $a\in\RR^3$. This ensures that $\hat y\in C^{0,1}(\Om,\RR^3)$.
\begin{pro}\label{prop3} Under the hypotheses on $\Om$ and $\Phi$ of Proposition \ref{prop2}, suppose $\hat y$ is given by \eqref{pw}. Then the cell-averaged energy satisfies
\be\label{pwe2}
 \bar E_\e\{\hat y,\Om\}= \int_\Om W(\nabla \hat y)dx+\e \int_\po \gamma(\nabla \hat y,n) dA+\e\int_\Ss \sigma(F^+,F^-, \hn)dA+o(\e),\ee
where $\gamma$ is the surface energy density given by \eqref{2}, and 
\be\label{ie2} \sigma(F^+,F^-, \hn)=\gamma(F^+, \hn)+\gamma(F^-, \hn)-\frac{2}{|F^+-F^-|}\int_{F^-}^{F^+}\gamma(F, \hn)dF\ee
is the interfacial energy density. The last term above is the\emph{ interaction energy}
\be\label{ie}\hat\sigma(F^+,F^-, \hn)=\frac{-2}{|F^+-F^-|}\int_{F^-}^{F^+}\gamma(F, \hn)dF=-2\int_{0}^{1}\gamma\left(tF^++(1-t)F^-, \hn\right)dt
\ee

\end{pro}
\begin{proof}
Consider $E(w,\e)$ as in \eqref{term}, but with $\hat y$ from \eqref{pw}. Pick and fix $w\in L$, and let $\Om^\pm_\e=-\e w+\Om^\pm$. The domain of integration $\Om\cap\Om_\e$ in \eqref{term} decomposes into four subdomains with disjoint interiors: $\Om^+\cap\Om^+_\e$, $\Om^-\cap\Om^-_\e$, $\Om^+\cap\Om^-_\e$ and $\Om^-\cap\Om^+_\e$. The integrals over the first two can be dealt with exactly as in Proposition \ref{prop2}, so that after summation over $w\in L$, they result in the first two terms in \eqref{pwe} below. We calculate the integral in \eqref{term} over the last two sets. 

Suppose first that $w\cdot \hn>0$. Then $\Om^+\cap\Om^-_\e=\emptyset$, and 
$\Om^-\cap\Om^+_\e$ differs from the oblique slab $\hat P_\e=\{x\in \RR^d: x=u+s w,\; u\in\Ss, \; -\e \le s\le 0 \}$ by a set of volume $O(\e^2)$. The part of the integral in \eqref{term} over this set is therefore
$$\int_{\hat P_\e}\Phi\left(\frac{\hat y(x+\e w)-\hat y(x)}{\e}\right)dx+O(\e^2)$$
 For $x\in \hat P_\e$, $\hat y(x+\e w)=F^+(x+\e w)$, $\hat y(x)=F^-x$. After use of \eqref{pw} and \eqref{hd}, since $u\cdot \hn=0$, the argument of $\Phi$ above is seen to be independent of $u$, and the domain $\hat P_\e$ can be transformed to $P_\e=\{x\in \RR^d: x=u+z \hn,\; u\in\Ss, \; -\e w\cdot \hn\le z\le 0 \}$ by a simple shear. The above integral becomes
\begin{gather*}|\Ss|\int_{ -\e w\cdot \hn}^0\Phi(\e^{-1} z a+F^+w)dz=\\ \e|\Ss| \int_0^1w\cdot \hn\Phi\left(F^-w+t(w\cdot \hn) a\right)dt =
-\e|\Ss| \int_0^1w\cdot \hn\Phi\left([tF^++(1-t)F^-]w\right) dt 
\end{gather*}
after setting $z=\e (t-1)(w\cdot \hn)$, changing the integration variable to $t$ and using \eqref{hd} to find $(w\cdot \hn) a=(F^+-F^-)w$.

 In case $w\cdot \hn=0$, both 
$\hat P_\e$ and $P_\e$ have measure zero, while the integral above vanishes.
The case $w\cdot \hn<0$ is similar and gives the negative of the above result. Combining these cases, summing over all $w\in L$, dividing by 2 and recalling the definition \eqref{2} we arrive at \eqref{pwe}, \eqref{ie}. These are trivially equivalent to \eqref{pwe2}, \eqref{ie2}.  \end{proof}

\begin{rem}\label{rem5} Note that \eqref{pwe2}, \eqref{ie2} are equivalent to
\be\label{pwe} \bar E_\e\{\hat y,\Om\}= \bar E_\e\{y^+,\Om^+\}+ \bar E_\e\{y^-,\Om^-\}+\e\int_\Ss\hat \sigma(F^+,F^-, \hn)dA+o(\e),\ee
where $y^\pm(x)=F^\pm x$ for $x\in\Om^\pm$, respectively, and
$$ \bar E_\e\{y^\pm,\Om^\pm\}=|\Om^\pm|W(F^\pm)+\e\int_{\Om^\pm}\gamma(F^\pm,n)dA+o(\e),$$
while
$\hat\sigma(F^+,F^-, \hn)$ is the\emph{ interaction energy} \eqref{ie}. Thus the total interfacial energy density $\sigma$ in \eqref{ie2} is equal to the sum of the surface energies of the two homogeneous deformations on either side of the interface, plus the interaction energy $\hat\sigma$. The main result here is the explicit expression for the interfacial energy \eqref{ie2} in terms of the surface energies.
\end{rem}
\begin{rem}
In view of \eqref{ie2}, the interfacial energy vanishes in the trivial case $F^+=F^-$, since the last (interaction) term (apart form the factor of $-2$) is the average of the surface energy density $\gamma(\cdot, \hn)$ over all convex combinations of $F^+$ and $F^-$. In view of the continuity of $\gamma(\cdot,n)$, this implies that the interfacial energy approaches zero as $|F^+-F^-|\to 0$, as in a model considered by \cite{ball}.
\end{rem}

\section{Discrete versus Cell-Averaged Energies}\label{sec6}

We compare cell-averaged energies with their fully discrete counterparts. Let $L=\ZZZ$. A \emph{crystallographic plane} is a plane that contains at least three non-collinear lattice points.. A plane is rational if it is parallel to a crystallographic plane. Rational planes have a \emph{ Miller normal}, one whose components are irreducible integers, known as the Miller indices of the plane. We assume that $\Om$ is a convex lattice polyhedron (the convex hull of lattice points). Then $\po$ is the union of convex lattice polygons. Each of these is part of a crystallographic plane. As a result, $\po$ has a Miller normal $\bn\in M$ almost everywhere (except edges), where
 \be\label{miller}M=\{\bn=(h,k,l)\in \ZZ^3:\rm{gcd}(h,k,l)=1\}\ee 
 We have the following three-dimensional version of \cite{rosakis}, Proposition 3.1 and case (a) of Proposition 5.1:
\begin{pro}\label{prop4} (i) Suppose $\Om\subset\RR^3$ is a convex lattice polyhedron with exterior Miller normal $\bn\in M$, $\nabla y=F=$const. on $\Om$ and $\Phi$ satisfies \eqref{ph}. Choose $\e_k=1/k$, $k\in\ZZ$ and let $L_k=\e_k L$.
 Then the discrete (unaveraged) energy \eqref{ene} satisfies
\be\label{11} E_\e\{y,\Om\}=\int_\Om W(F )dx+\e \int_\po \gamma_\diamond(F,\bn) dA+o(\e_k) \quad\hbox{ for } \e= \e_k=1/k,\;\;k\in\ZZ,\;\; k\to\infty,\ee
where the stored energy function $W$ is given by \eqref{0} and the surface energy density $\gamma_\diamond$ by
\be\label{12}\gamma_\diamond(F,\bn)=-\frac{1}{4}\sum_{w\in \Lo} \frac{1}{|\bn|}\bigl(|w\cdot \bn|-1\bigr)\Phi(Fw)= \gamma\bigl(F,\frac{1}{|\bn|}\bn\bigr)+\frac{1}{2|\bn|}W(F)\ee
for $\bn\in M$ with $\gamma$ as in \eqref{2}. \hfill\break
\noindent(ii) For each $k\in\ZZ$ there is $\Om_k$ such that 
\be\label{ok}\Om_k\cap L_k=\Om\cap L_k,\quad|\Om_k| - \e_k^3\#(\Om_k\cap L_k)=o(\e_k).\ee
There holds
\be\label{13}\bar E_\e\{y,\Om_k\}=\int_{\Om_k} W(F )dx+\e\int_{\po_k} \gamma(F,n) dA+o(\e_k)\quad\hbox{ for } \e= \e_k=1/k,\;\;k\in\ZZ,\;\; k\to\infty,\ee
where $n=\bn/|\bn|$ on $\po$ and $\gamma$ is as in \eqref{2}. 
 \end{pro}
 %\begin{rem} 
The bulk elastic energy \eqref{lim} is identical in its fully discrete and cell-averaged incarnations (Theorem 1 of \bbl \cite{blanc1} and Proposition 1 above). The surface energy density $\gd$ in \eqref{11}, \eqref{12} differs from the cel-averaged version $\gamma$ in \eqref{2} by the additive term $\frac{1}{2|\bn|}W(F)$. Because of this term, $\gd(F,\cdot)$, defined only for rational normals $\bn\in M$, cannot be continuously extended to irrational normals (Proposition 4.4 \cite{rosakis}). One way to cure this pathology is part (ii) of Proposition \ref{prop4}, \emph{cf}. Proposition 5.1in  \cite{rosakis}.
The problem is due to the geometric discrepancy \eqref{rem} between the continuous and discrete volumes, which for lattice polyhedra is of order $O(\e)$, the same order as the surface energy itself.

To correct this, it is possible to modify  $\Om$ to correct the discrepancy; see case (ii) in Proposition \ref{prop4}. One does this by exploiting the space between crystallographic planes.  The distance between two adjacent such planes with Miller normal $\bn\in M$ is $1/|\bn|$.  By pushing  each facet of $\Om$ outwards by half that distance one produces the augmented domain $\Om_k$ whose discrepancy \eqref{ok} is of order higher than $O(\e)$.   The appropriate surface energy density for this domain, in the sense of \eqref{13}, is $\gamma$, which is free of discontinuities. It also coincides with the cell-averaged one \eqref{2}.  
The reason for this  coincidence is that  by taking cell averages, in effect we eliminate the geometric discrepancy as shown by Corollary \ref{cor1}.

% \end{rem}
Next we compute the discrete energy of a coherent phase boundary. Let the interface $\Ss$ be the intersection of a crystallographic plane with $\Om$, that separates it into two parts $\Om^\pm$ with nonempty interior. Then $\Ss$ is a lattice polygon containing the origin $0$ with $\Omp\cup\Omm=\bar\Om$ and $\Omp\cap\Omm=\Ss$, so that $\Ompm$ are closed  convex lattice polyhedra. 

\begin{pro}\label{prop5}  Suppose $\Om\subset\RR^3$ is a convex lattice polyhedron with exterior Miller normal $\bn\in M$, $\Ss=P\cap\Om$ where $P$ is a crystallographic plane with Miller normal $\tn\in M$, $\hat y$ is given by \eqref{pw} and $\Phi$ satisfies \eqref{ph}. Choose $\e_k=1/k$, $k\in\ZZ$ and let  $L_k=\e_k L$.
 Then the discrete (unaveraged) energy \eqref{ene} satisfies
 \be\label{51} E_\e\{y,\Om\}=\int_\Om W(\nabla\hat y )dx+\e \int_\po \gamma_\diamond(\nabla\hat y ,\bn) dA+\e\int_\Ss \tau(F^+,F^-, \tn)dA+o(\e), \quad  \e= \e_k=1/k,\;\;k\in\ZZ,\;\;  k\to\infty,\ee
where $\gamma_\diamond$ is the surface energy density given by \eqref{12}, and 
\be\label{52} \tau(F^+,F^-, \tn)=\gamma(F^+, \hn)+\gamma(F^-, \hn)+\hat \tau(F^+,F^-, \tn)\ee
is the interfacial energy density. The last term above is the interaction energy:
\be\label{inten}\hat\tau(\Fp,\Fm,\tn) =\frac{1}{ 2|\tn|}\sum_{w\in L} \; \left[\frac{1}{ 2}\Phi(\Fm w)+\frac{1}{ 2}\Phi(\Fp w)+
\sum_{j=1}^{|w\cdot \tn| -1} \Phi\bigl( \Fm w +\frac{j}{|w\cdot \tn|}(\Fp-\Fm)w \bigr)\right] 
\ee
for $\tn\in M$. 
\end{pro}
 
 We postpone the proof  of Proposition \ref{prop5} until the end of this section. 

In the one-dimensional case, the analogous interfacial energy is obtained by Mora-Corral (\cite{mora}  Theorem 6).  Blanc \& Le Bris \cite{blanc2} consider a model in three dimensions where atomic interactions in an interfacial layer of finite thickness are harmonic. 
 
The fully discrete surface interfacial energy density with its cell-averaged counterpart appear quite different at first glance. The difference occurs in the interaction terms $\hat \tau$ of \eqref{inten} versus $\hat\sigma$ of \eqref{ie}.
 On the other hand,  observe that (modulo a multiplicative constant) the term in brackets is formally identical to the \emph{trapezoidal approximation} of the line integral
$\int_{\Fm }^{\Fp } \Phi(Fw) dF$ over a straight path with endpoints $\Fpm $, partitioned into $|w\cdot \tn|$ subintervals. For a function $f:\RR^n\to \RR$, $a$, $b\in \RR^n$, $K\in\ZZ_+$, the \emph{trapezoidal sum} of the function $f$ from $a$ to $b$ with partition number $K$ is
\be\label{trs}T[f;a,b,K]={|b-a|\over K}\sum_{j=0}^{K-1}{1\over2}\bigl[ f(x_j)+f(x_{j+1})\bigr],\quad x_j=a+{j\over K}(b-a),\quad j=0,\ldots,K\ee
In addition, both  discrete densities $\gd(F,\cdot)$ and $\hat\tau(\Fp,\Fm,\cdot)$ are defined for \emph{rational} normals.  If we approximate an irrational normal by a sequence of rational ones $\bn_j\in M$, then necessarily $|\bn_j|\to\infty$ and in this limit, it turns out that the difference between discrete and cell averaged densities disappears. Actually,  rational orientations can be approximated by such sequences.  This clarifies the connection between discrete and cell averaged energies:

\begin{pro}\label{prop6}  Given any unit vector $n\in S^2$, there is a sequence of vectors $\bn_j$, $j\in\ZZ$,  such that
\be\label{seq}\bn_j\in M, \quad \frac{1}{|\bn_j|}\bn_j\to n, \quad |\bn_j|\to\infty \;\;\hbox{as }j\to\infty.\ee
For such a sequence of Miller normals, the corresponding surface and interfacial energy densities approach their cell-averaged counterparts in the limit as $j\to\infty$. In particular,
\be\label{lmtg}\lim_{j\to\infty}\gd(F,\bn_j)=\gamma(F,n).\ee
Also, for any sequence $F^\pm_j\to\Fpm$ such that $F^+_j-F^-_j=a_j\otimes\hn_j$ for some $a_j\to a\in\RRR$, with $\hn_j=\bn_j/|\bn_j|$,
\be\label{lmts} \lim_{j\to\infty}\tau(F^+_j,F^-_j,\bn_j)=\sigma(\Fp,\Fm,n)\ee
where $\sigma$ is defined in \eqref{ie2} and $\tau$  in \eqref{52}, \eqref{inten}.
\end{pro}

\begin{proof}  If $n\in S^2$ is \emph{irrational}, so that there is no $m\in M$ with $n=m/|m|$, the existence of a sequence \eqref{seq} is easy. Construct a convergent sequence of rational approximations of the components of $n$ and multiply each
resulting vector with the least common multiple of the 3 denominators to obtain $\bn_j$.

 If  $n\in S^2$ is \emph{rational}, so that $n=m/|m|$ for some $m\in M$, then there is a lattice ($\ZZZ$) basis of vectors $d_i\in M$, $i=1,2,3$,  such that  $d_1\cdot m=d_2\cdot m=0$, $d_3\cdot m=1$. In fact , letting  $b_i\in M$, $i=1,2,3$ be the dual basis vectors, so that $b_i\cdot d_k=\dd_{ik}$, one has $m=b_3$.  Then one may choose $\bn_j=b_1+b_2+jb_3$ for $j\in\ZZ$.  Clearly the last two conditions of \eqref{seq} hold. To show that $\bn_j\in M$, choose $p_j=-jd_1+d_2+d_3\in\ZZZ$ and note that  $\bn_j\cdot p_j=1$ for $j=1,2,3,\dots$, so  by Bezout's Lemma, the components of $\bn_j$ are coprime and $\bn_j\in M$.

Having established \eqref{seq}, replace $\bn$ by $\bn_j$ in \eqref{12}, and  the limit in \eqref{lmtg} is trivial by \eqref{seq}.

In view of \eqref{trs},  letting $\Phi_w(F)=\Phi(Fw)$,  \eqref{inten} can be written as
 \begin{align}\label{tau}
\hat \tau(\Fp,\Fm,\tn)=&{1\over 2|\tn||\Fp-\Fm|}\sum_{w\in L} |w\cdot \tn|T\!\left[\Phi_w(\cdot),\Fm ,\Fp ,|w\cdot \tn|\right]\notag\\
=&{1\over 2|\Fp-\Fm|}\sum_{w\in L} |w\cdot \hn|T\!\left[\Phi_w(\cdot),\Fm ,\Fp ,|w\cdot \tn|\right] 
\end{align}
where $\hn=\tn/|\tn|$.  Choose the limit vector $n=\hn$ in \eqref{seq}. 
Replace $\tn$ by $\bn_j$, $\hn$ by $\hn_j$ and $\Fpm$ by $F^\pm_j$ above and note that unless $w\cdot \hn=0$ (trivial case), we have that $ |w\cdot \bn_j|\to\infty$ as $j\to\infty$.
The trapezoidal sum $T\!\left[\Phi_w(\cdot),F^-_j ,F^+_j ,|w\cdot \tn_j|\right] 
$  then converges to the corresponding integral
$\int_{\Fm }^{\Fp } \Phi(Fw)dF$.
From \eqref{tau} and  \eqref{2},  $\hat\tau(F^+_j,F^-_j,\bn_j)$ clearly converges to $\hat\sigma(\Fp,\Fm,n)$ of \eqref{ie}.\end{proof}  

\begin{rem}  The reason for the convergence of discrete to cell-averaged energies in the ``large Miller index'' limit  \eqref{seq} is related once again to the geometric discrepancy  \eqref{rem}.  One may approximate a lattice polyhedron $\Om$ by a sequence of rational polyhedra $\Om_j$ whose facets have Miller normals that converge to those of $\Om$ in the sense of \eqref{seq}. Then one can show that the remainder \eqref{rem} for $\Om_j$ tends to zero for large $j$. This is related to the fact that the interplanar spacing between crystallographic planes with Miller normal $\bn$ is $1/|\bn_|$. 

If we approximate both the outside Miller normal $\bn$ and the interfacial one $\hn$ by means of large index sequences as in \eqref{seq}, then the entire discrete energy \eqref{51} converges to its cell-averaged counterpart \eqref{pwe2}.
\end{rem} 
\noindent \emph{Proof of Proposition \ref{prop5}.} Consider the piecewise homogeneous deformation \eqref{pw}. Choose $\e_k=1/k$, $k\in\ZZ$ and let  $L_k=\e_k L$, $E_k=E_{\e_k}\{\hat y,\Om\}$  in \eqref{ene}.
We split the energy as follows. First we split the outer sum in \eqref{ene}:
 $$E_k=E^+_k+E^-_k - E^\Ss_k;$$
 where, letting $\Phi$ stand for $\Phi(k(\hat y(z)-\hat y(x)))$, 
 \be\label{es}E^\pm_k={1\over2k^3}\sum_{x\in \lk\cap \Ompm} \; \sum_{z\in \lk\cap \Om} \Phi,\quad 
 E^\Ss_k={1\over2k^3}\sum_{x\in \lk\cap \Ss} \; \sum_{z\in \lk\cap \Om} \Phi.\ee
 The minus sign precludes $\Ss$ from contributing to both $E^\pm_k$.
 We further split $E^\pm_k$ by splitting the inner sum above:
 \be\label{splitt}E^+_k=E_k^{++}+E^{+-}_k,\quad E^-_k=E_k^{--}+E^{-+}_k.\ee
 where
 \be\label{dir}E_k^{++}={1\over2k^3}\sum_{x\in \lk\cap \Omp} \; \sum_{z\in \lk\cap \Omp} \Phi(k\Fp(z- x))=E_{\e_k}\{\hat y,\Om^+\}\ee
 is the energy of $\Omp$, that is, with $\Omp$ replacing $\Om$ and $y(x)=\hat y(x)=\Fp x$ in \eqref{11}, while
 \be\label{eint} E_k^{+-}={1\over2k^3}\sum_{x\in \lk\cap \Omp} \; \sum_{z\in \lk\cap \Om\bc\Omp} \Phi(k(\Fm z-\Fp x))\ee
 is part of the \emph{interaction energy} between $\Omp$ and $\Omm\bc \Ss=\Om\bc\Omp$. The total interaction energy $E_k^{+-}+E^{-+}_k$ with $+$ and $-$ interchanged in \eqref{eint} for the second term.
 The remaining terms in $E_k$ are $E^{++}_k+E^{--}_k - E^\Ss_k$, which we consider first. The restriction of $\hat y$ to $\Omp$ and $\Omm$ is an affine deformation, while $\Ompm$ are convex lattice polyhedra with respect to $L_k$.  Proposition \ref{prop4}  then applies to $E^{++}_k$ with $F=\Fp$ and $\Om=\Omp$, and similarly for $E^{--}_k$, and asserts that 
$$E^{\pm\pm}_k=\int_{\Om^\pm} W(F^\pm)dV+\kk \int_{\partial\Om^\pm}{\gamma_\diamond}(F^\pm,\bn)dA+o(1/k)
$$
as $k\to\infty$. Adding the $\pm$ contributions,  we obtain
\be\label{mess}E^{++}_k+E^{--}_k =\int_\Om W(\nabla\hat y)dV+\kk \int_{\partial\Om}{\gamma_\diamond}(\nabla\hat y,\bn)dA+\kk \int_{\Ss}[{\gamma_\diamond}(\Fp,\tn)+{\gamma_\diamond}(\Fm,\tn)]dA+o(1/k) ,
 \ee
 where $\tn$ is the Miller normal to $\Ss$. From \eqref{es},
 \begin{gather*}E^\Ss_k={1\over2k^3}\sum_{x\in \lk\cap \Ss} \; \sum_{z\in \lk\cap \Om}\Phi(k(\hat y(z)-\hat y(x))) \\
= {1\over2k^3}\sum_{x\in \lk\cap \Ss} \left[ \sum_{z\in \lk\cap \Omp}\Phi(k\Fp(z- x))+ \sum_{z\in \lk\cap \Omm}\Phi(k\Fm(z- x))-\sum_{z\in \lk\cap \Ss}\Phi(k\Fpm(z- x))\right] \\
 = {1\over2k^3}\sum_{x\in \lk\cap \Ss} \; \left[ 
 \sum_{w\in L_+}\Phi(\Fp w )+ \sum_{w\in L_-}\Phi( \Fm w )-\sum_{w\in L_0}\Phi( \Fpm w )\right] \\ 
 - {1\over2k^3}\sum_{x\in \lk\cap \Ss} \; \sum_{z\in \lk\bc \Om}\Phi(k(\hat y(z)-\hat y(x)))
 \end{gather*}
 where in the third line, $w=k(z-x)$, $L_\pm=\{x\in L: \pm x\cdot\tn\ge 0\}$, $L_0=\{x\in L: x\cdot\tn= 0\}$, we have extended the definition of $\hat y$ in \eqref{pw} to all of $\RRR$ by letting $\hat y(x)=\Fpm x$ for $\pm x\cdot\tn\ge 0$, $x\in\RRR$. Note also that $\Fp x=\Fm x$ for $x\cdot\tn= 0$. The last integral above is $o(1/k)$. The sums over $w\in L_\pm$ in the third line above can be extended over the whole of $L$ after division by $2$, the summands being even in $w$. The fact that the contribution of $w\in L_0$ is counted twice is rectified by the third term in brackets. The inner sums are independent of $x$, hence the outer sum over $x\in \Ss$ gives a factor of $\#(\Ss\cap\lk)$.  Since $\Ss$ is a lattice polygon, it is part of a crystallographic plane containing a two-dimensional lattice with unit cell area equal to the norm $|\tn|$ of the  Miller normal of $\Ss$, e.g., Beck \& Robins \cite{br}.  As a result we estimate 
 $$\#(\Ss\cap\lk)= k^2|\Ss|/|\tn|+O(k).$$
 After recalling \eqref{0}, the result is
 $$
 E^\Ss_k={1\over 2k} \int_{\Ss}{W(\Fp)+W(\Fm)\over |\tn|}dA+o(1/k).
$$ 
Combine this with \eqref{mess} and recall \eqref{2}, \eqref{12} to conclude
\be\label{messs}E^{++}_k+E^{--}_k- E^\Ss_k= \int_\Om W(\nabla\hat y)dV+\kk \int_{\partial\Om}{\gamma_\diamond}(\nabla\hat y,\bn)dA+\kk \int_{\Ss}[{\gamma}(\Fp,\hn)+{\gamma}(\Fm,\hn)]dA+o(1/k) .
 \ee
where $\hn=-\tn/|\tn|$.  The integral over the interface $\Ss$ involves the  surface energy density $\gamma$ from  \eqref{2}, while those over the boundary involve $\gamma_\diamond$ of  \eqref{12}.
 
 We turn to the interaction energy $E_k^{+-}$ from \eqref{eint}. As in \bbl \cite{blanc1}, we introduce a parameter $\dd=\dd_k>0$, and write \eqref{eint} as
 $$ E_k^{+-}={1\over2k^3}\sum_{x\in \lk\cap \Omp} \; \sum_{z\in \lk^-\cap B_\dd(x)} \Phi(k(\Fm z-\Fp x))
 +O((k\dd)^{-p}) +O(\dd^2/k).$$
 letting $L^-_k=\{z\in L_k : z\cdot \hn <0\}$. After extending the deformation so that  \eqref{pw} holds  in the whole of $\RRR$, the $O((k\dd)^{-p})$ term in  the remainder is estimated using Lemma 3.2 of \cite{rosakis}, since $\Phi$ satisfies \eqref{ph}, in the same spirit as in the proof of Proposition 3.3 in \cite{rosakis}. The  $O(\dd^2/k)$  estimate is due to the fact the inner sum over $z$  is bounded, and differs from the one over  $ \lk\cap (B_\dd(x)\bc\Omp)\cap\Omm$ only for a set of $x$ near $\Ss\cap\po$ and having  measure $O(\dd^2)$,  therefore containing $O(\dd^2k)$ elements (then dividing by $k^3$).  The only nonzero contributions to the sum above are from
$x\in\Omp$ with $\rm{dist}(x,\Ss)<\dd$. Let $\hn=-\tn/|\tn|$,
 $$ H_{\dd}=\{x\in\RRR: 0\le  x\cdot \hn<\dd\},\quad H^+_\dd=H_\dd\cap\Omp$$ 
Then,
 \be\label{pri} E_k^{+-}={1\over2k^3}\sum_{x\in \lk\cap H^+_\dd} \; \sum_{z\in \lk^-\cap B_\dd(x)} \Phi(k(\Fm z-\Fp x)) 
 +O((k\dd)^{-p}) +O(\dd^2/k).\ee
 We let $w'=z-x$ and reverse the order of summation. The sum above becomes
\be\label{tsir}{1\over2k^3}\sum_{w'\in \lk\cap B^+_{\dd}}\;  \sum_{x\in \lk\cap P_{w'}} \Phi\bigl(k( \Fm w'+(x\cdot \tn){a\over |\tn|}
)\bigr) +O(\dd^3)\ee
where
 $$ P_{w'}=\{x\in\RRR: x=x_\Ss +u,\; x_\Ss\in\Ss, \; u\cdot\tn=0, 0\le -x\cdot\tn<w'\cdot\tn\},\quad  
B^+_{\dd}=\{w'\in B_\dd(0): w'\cdot \tn>0\}$$
and we have used \eqref{hd}. The second sum differs from the one in  \eqref{pri} by a portion near the boundary of $\Ss$ of measure $O(|w'|^2)$ and the summand is bounded, hence the estimate $O(\dd^3)$ above. Change variables to $w=kw'$, $j=-kx\cdot\tn$. Observe that for $x\in L$, $x\cdot\tn$ is an integer, hence so are $j$ and $w\cdot\tn$. Write the sum over $x$ as a double one over $x_\Ss\in \Ss$ and $0\le j<w\cdot \tn$. The summand depends on $w$ and $j$ but not on $x_\Ss$, hence summation over the latter variable gives a factor of $\#(\Ss\cap\lk) = k^2|\Ss|/|\tn|+O(k)$. The previous sum becomes
\be\label{is1}{|\Ss|\over 2k|\tn|}\sum_{w\in L\cap B^+_{k\dd}} \sum_{0\le j<w\cdot \tn } \Phi( \Fm w- {j\over |\tn|}a) +O(1/k^2)\ee
The analogous contribution from $E^{-+}_k$ is 
\be\label{is2}{|\Ss|\over 2k|\tn|}\sum_{w\in L\cap B^-_{k\dd}} \sum_{0\le j<w\cdot (-\tn) } \Phi(\Fp w- {j\over |\tn|}a),\ee
where $B^-_{k\dd}=\{w\in B_{k\dd}(0): w\cdot \hn<0\}$, $j=-x\cdot(-\tn)$ and we have written $\Fm-\Fp=a\otimes(-\hn)$.
In order to combine the sums from \eqref{is1} and \eqref{is2}, note that letting 
$K=|w\cdot\tn|$, 
 \begin{align*}
 \Fm w- {j\over |\tn|}a=\Fm w +{j\over K}(\Fp-\Fm)w & =(1-{j\over K})\Fm w+{j\over K}\Fp w,\quad w\cdot\tn >0,\notag\\
 \Fp w-{j\over |\tn|}a=\Fp w +{j\over K}(\Fm-\Fp)w & ={j\over K}\Fm w+(1-{j\over K})\Fp w ,\quad w\cdot\tn <0. 
 \end{align*}
Thus in adding \eqref{is1} and \eqref{is2} we may combine the $j$th term from the first sum to the $K-j$th term from the second for $j=1,\ldots,K-1$, yielding a term summed over $w\in B_{k\dd}(0)$. Note here that for $w\cdot \tn=0$ both sums vanish. The $j=0$ terms in both sums do not combine, but the summands are $\Phi(\Fpm w)$ which are both even in $w$ so we sum them over $w\in B_{k\dd}(0)$ after dividing by 2. The result is
$${| \Ss|\over 2k|\tn|}\sum_{w\in L\cap B_{k\dd}(0)} \; \left[{1\over 2}\Phi(\Fm w)+{1\over 2}\Phi(\Fp w)+
\sum_{j=1}^{|w\cdot \tn| -1} \Phi\bigl({j\over |w\cdot \tn|}\Fm w+(1-{j\over |w\cdot \tn|})\Fp w \bigr)\right] $$
Extending the outer sum to the whole of $L$ leaves a remainder which we estimate as before using \eqref{ph} by
\be\label{tsiu}
C\kk\sum_{w\in L\bc B_{k\dd}(0)} |w\cdot n| |\Phi(\lambda w)|<\kk\sum_{w\in L\bc B_{k\dd}(0)} C|w | ^{-(2+p)}<Ck^{-p}\dd^{1-p}.
\ee
 The remainders in \eqref{pri}-\eqref{is1} and  \eqref{tsiu} can be rendered $o(1/k)$ by choosing $p$ and $\dd = k^{-1+\zeta}$ suitably.  The total interaction energy is therefore
$$E_k^{+-}+E_k^{-+}=\kk|\Ss| \hat\tau(\Fp,\Fm,\tn) +o(1/k) $$
where the \emph{interaction energy density} is given by 
$$\hat\tau(\Fp,\Fm,\tn) ={1\over 2|\tn|}\sum_{w\in L} \; \left[{1\over 2}\Phi(\Fm w)+{1\over 2}\Phi(\Fp w)+
\sum_{j=1}^{|w\cdot \tn| -1} \Phi\bigl( \Fm w +{j\over|w\cdot \tn|}(\Fp-\Fm)w \bigr)\right] 
$$
for $\tn\in M$; see \eqref{miller}.  In view of \eqref{splitt} and  \eqref{messs}, \eqref{51}-\eqref{inten} are confirmed.

\subsection*{Acknowledgments} Thanks are due especially to J.R. Willis for his generous hospitality and stimulating discussions on periodic media which inspired some of the key ideas of this work at DAMTP, Cambridge; also ACMAC for support and MF Oberwolfach for hospitality. L. Truskinovsky, D. Mitsoudis and C. Makridakis for discussions.
%\end{acknowledgements}

\bibliographystyle{spmpsci}
\bibliography{refinterfac}
 \end{document}